\documentclass[11pt]{article}

\usepackage[a4paper,hmargin=1.3in,vmargin=1.2in]{geometry}

\usepackage{amsmath,amsthm,amssymb,amsfonts,verbatim}
\usepackage{fullpage}
\usepackage{todonotes}
\usepackage[pagebackref]{hyperref}
\usepackage{arydshln}

 \providecommand{\F}{\mathbb{F}}

\date{}

\newtheorem{lem}{Lemma}[section]
\newtheorem{thm}{Theorem}
\newtheorem{cor}{Corollary}

\newtheorem{rem}{Remark}
\newtheorem{exm}{Example}

\theoremstyle{remark}

\def\F {{\mathbb{F}}}
\def\Fs {{\mathbb{F}_{q^2}}}

\def\cP{{\mathcal P}}

\def\GRS{{\mathcal GRS}}
\def\ba{{\bf a}}
\def\bb{{\bf b}}
\def\bc{{\bf c}}

\def\bu{{\bf u}}
\def\bv{{\bf v}}
\def\bw{{\bf w}}
\def\bo{{\bf 0}}
\def\bi{{\bf 1}}

\def\bx{{\bf x}}

\def\Pin{{\infty}}
\def\a{{\alpha}}
\def\Ga{{\alpha}}
\def\Gb{{\beta}}

\def\beq{\begin{equation}}
\def\eeq{\end{equation}}
\begin{document}

\title{A Construction of New Quantum MDS Codes }

\author{Lingfei~Jin and Chaoping Xing
\thanks{The authors are with Division of Mathematical Sciences, School
of Physical and Mathematical Sciences, Nanyang Technological
University, Singapore 637371, Republic of Singapore
(email: \{lfjin,xingcp\}@ntu.edu.sg).}
 \thanks{The work was partially supported by the Singapore  Ministry of Education  Tier 1 grant 01/14.}
 }

\maketitle

\begin{abstract} It has been a great challenge to construct new quantum MDS codes. In particular, it is very hard to construct quantum MDS codes with relatively large minimum distance. So far, except for some sparse lengths, all known $q$-ary quantum MDS codes have minimum distance less than or equal to $q/2+1$. In the present paper, we provide a construction of  quantum MDS codes with minimum distance bigger than $q/2+1$. In particular, we show existence of $q$-ary quantum MDS codes  with length $n=q^2+1$ and minimum distance $d$ for any $d\le q+1$ except for $d=q$ (this result extends those given in \cite{Gu11,Jin1,KZ12}); and with length $(q^2+2)/3$ and minimum distance $d$ for any $d\le (2q+2)/3$ if $3|(q+1)$. Our method is through Hermitian self-orthogonal codes. The main idea of  constructing Hermitian self-orthogonal codes is based on the solvability in $\F_q$ of a system of homogenous equations over $\Fs$.
\end{abstract}

\section{Introduction}
The field of quantum error correction has experienced a great progress since the establishment of the connections between quantum codes and classical codes (see \cite{Cal Rai}). One of these connections shows that
the construction of  quantum codes
can be reduced to that of classical linear error-correcting codes with certain self-orthogonality properties (see \cite{Ash Kni,Cal
Rai, Ket Kla,Shor}). More precisely, we can construct  quantum codes as long as we can construct
classical linear codes with symplectic, Euclidean or Hermitian
self-orthogonality (see \cite{Aly Kla Sar,Che Lin Xin,Li Xin Wan2,Sar Kla,Stea3}, etc). The quantum codes obtained by using self-orthogonality are called stabilizer codes.

 For a prime power $q$, a $q$-ary $((n,K,d))$  quantum code is a $K$ -dimensional
vector subspace of the Hilbert space
$\left(\mathbb{C}^q\right)^{\otimes n}$ which can detect up to
$d-1$ quantum errors, or equivalently, correct up to
$\lfloor(d-1)/2\rfloor$ quantum errors. If we put $k=\log_q K$, we denote a $q$-ary $((n,K,d))$  quantum code by $[[n,k,d]]_q$ or simply $[[n,k,d]]$ if the $q$ is clear in the context.  Similar to the classical Singleton bound, the parameters of an $[[n,k,d]]_q$ quantum code have to satisfy the quantum Singleton bound:
$k\leq n-2d+2$. A quantum code achieving this quantum Singleton bound is called a
quantum maximum-distance-separable (MDS) code.

Just as in the classical case, quantum MDS codes form an optimal family of quantum codes. Constructing quantum MDS codes has become one of the central topics for quantum codes in the recent years. The length of a $q$-ary quantum stabilizer MDS code cannot exceed $q^2+1$ if we assume the classical MDS conjecture holds. In fact, several families of quantum MDS codes have been constructed in the past few years (see \cite{Bie Ede,Feng,Gra Bet2,Hu Tan,Li Xin Wan2,Sar Kla, Jin1,Gu11,KZ12}). The problem of constructing quantum MDS codes with $n\le q+1$ has been completely solved  through classical Euclidean self-orthogonal codes \cite{Gra Bet2, Gra2}. However, constructing quantum MDS codes with $n>q+1$ is a more challenging task. In particular, it is very hard to construct quantum MDS codes with relatively large minimum distance. So far, except for some sparse lengths $n$ such as $n=(q^2+1)/2$, $q^2$ and $q^2+1$, all known $q$-ary quantum MDS codes have minimum distance less than or equal to $q/2+1$ (see \cite{Bie Ede,Feng,Gra Bet2,Gra2,Hu Tan,Sar Kla,Gu11,KZ12}).

 In this paper, we construct some new quantum MDS codes with minimum distance bigger than $q/2+1$ through classical Hermitian self-orthogonal generalized Reed-Solomon codes. In particular, we show existence  of $q$-ary quantum MDS codes  with length $n=q^2+1$ and minimum distance $d$ for any $d\le q+1$ (this result extends those given in \cite{Gu11,Jin1,KZ12}); and with length $(q^2+2)/3$ and minimum distance $d$ for any $d\le (2q+2)/3$ if $3|(q+1)$.
\subsection*{Main result and comparison with previous constructions}
Previously, the known $q$-ary quantum MDS codes with lengths $n$ satisfying either $n=(q^2+1)/2$ , $q^2$ and $ q^2+1$ may have minimum distance bigger than $q/2+1$. More precisely, there exist $q$-ary $[[n,n-2d+2,d]]$  quantum MDS codes for the following $n$ and $d$:
\begin{itemize}
  \item[(i)] $n=q^2+1$ and $d=q+1$ (see \cite{Li Xin Wan2}); and $n=q^2+1$ and $ d\le q+1$ for even $q$ and odd $d$  (see \cite{Gu11}); and $n=q^2+1$ and $d\le q+1$ for $q\equiv 1\mod{4}$ and even $d$ (see \cite{KZ12}).
       \item[(ii)] $n= q^2$ and $d\leq q$ (see \cite{Gra Bet2} and \cite{Li Xin
  Wan2,Jin1}).
  \item[(iii)] $n=(q^2+1)/2$ and $q/2+1< d\le q$ for odd $q$ (see \cite{KZ12}).
\end{itemize}
Based on (ii) and (iii) above, by using a propagation rule \cite{FLX06}, one can obtain the following quantum MDS codes with minimum distance bigger than $q/2+1$.

\begin{itemize}
  \item[(iv)] $n=q^2-s$ and $q/2+1<d\le q-s$ for $0\le s<q/2-1$.
  \item[(v)] $n=(q^2+1)/2-s$ and $q/2+1<d\le q-s$ for $0\le s<q/2-1$.
\end{itemize}
Our paper demonstrates new $q$-ary quantum MDS codes with minimum distance bigger than $q/2+1$ for lengths $n=q^2+1$ and $n=r(q-1)+1$ with $q+1\equiv r\mod{2r}$. More specifically, we have the following main result in this paper.
\begin{thm}[Main Theorem]\label{thm:1.1} There exists $q$-ary $[[n,n-2d+2,d]]$ quantum MDS codes for the following $n$ and $d$.\begin{itemize}
\item[{\rm (i)}] $n=q^2+1$ and $d\le q-1$ or $d=q+1$.
  \item[{\rm (ii)}] $n=r(q-1)+1$ and $d\le (q+r+1)/2$ for $q\equiv r-1\mod{2r}$. In particular, if $3$ divides $q+1$, then there exists a $q$-ary quantum MDS codes of length $n=(q^2+2)/3$ and minimum distance $d$ for any $d\le (2q+2)/3$.
\end{itemize}
\end{thm}
\begin{proof} Part (i) is just Theorem \ref{thm:4.1}.

For Part (ii), we must have $r|(q+1)$. The case of $r=q+1$ follows from Example \ref{exm:4.5}. The case of $r<q+1$ follows from Theorem \ref{thm:4.2}.
\end{proof}

\begin{rem}{ \rm Using a propagation  rule \cite{FLX06}, we can derive more quantum MDS codes of other lengths from Theorem \ref{thm:1.1}(ii).
}
\end{rem}

\subsection*{Our techniques} Our idea of constructing quantum codes is through the construction of Hermitian self-orthogonal generalized Reed-Solomon (GRS or generalized RS for short) codes. However, unlike the construction of Euclidean self-orthogonal $\GRS$ codes, it is a big challenge to construct Hermitian self-orthogonal $\GRS$ codes due to the fact that $q$th power has to be raised for the Hermitian inner product. To solve this problem, the first step is to make sure that a certain system of homogenous equations over $\Fs$ has solutions over $\F_q$. We first provide a necessary and sufficient condition under which this system of homogenous equations over $\Fs$ has solutions over $\F_q$. The second step is to find some suitable evaluation points.

\subsection*{Organization of the paper}
In Section 2, we first study a system of homogenous equations over $\Fs$ and give a sufficient and necessary condition under which this system has a nonzero solution over $\F_q$. Afterwards, we show that these conditions are satisfied in some cases and consequently the system has nonzero solutions over $\F_q$ for theses cases. In Section 3, we introduce Hermitian self-orthogonal codes and provide some sufficient conditions under which a GRS code is Hermitian self-orthogonal. We also provide some examples of Hermitian self-orthogonal MDS codes. Finally, we present some quantum MDS codes based on Hermitian self-orthogonal MDS codes in Section 3.

\section{A system of equations}
Let us consider a system of equations over $\F_{q^2}$ given by
\begin{equation}\label{eq:2.1}
A\bx^T=\bo,
\end{equation}
where $A$ is an $(n-1)\times n$ matrix of rank $n-1$ over $\Fs$ and $T$ stands for transpose. One knows that (\ref{eq:2.1}) must have at least one nonzero solution over $\Fs$. However, for our application, we are curious about the question that  whether (\ref{eq:2.1}) has a nonzero solution over $\F_q$. In this section, we give some sufficient and necessary conditions under which (\ref{eq:2.1}) has a nonzero solution over $\F_q$.

\begin{lem}\label{lem:2.1} The equation (\ref{eq:2.1}) has a nonzero solution in $\F^n_q$ if and only if $\bc^q$ is a solution of (\ref{eq:2.1}) whenever $\bc$ is a solution of (\ref{eq:2.1}).
\end{lem}
\begin{proof} If (\ref{eq:2.1}) has a nonzero solution $\bb$ in $\F^n_q$, then the solution space of (\ref{eq:2.1}) is $\Fs\cdot\bb=\{\a \bb:\; \alpha \in\Fs\}$ since the solution space has dimension $1$ over $\Fs$. Thus, for every solution $\lambda\bb$, we have $(\lambda\bb)^q=\lambda^q\bb\in \Fs\cdot\bb$.

Conversely, assume that $\bc^q$ is a  solution of (\ref{eq:2.1}) for a nonzero solution $\bc$  of (\ref{eq:2.1}). Choose a basis $\{1,\a\}$ of $\Fs$ over $\F_q$. Consider the two elements $\bw_1:=\bc+\bc^q$ and $\bw_2:=\a\bc+\a^q\bc^q$. It is clear that both $\bw_1$ and $\bw_2$ are solutions of (\ref{eq:2.1}) in $\F^n_q$. On the other hand, we have
\[\left(\begin{array}{c}
\bc \\
\bc^q\end{array} \right)=\left(\begin{array}{cc}
1 & 1 \\
\a & \a^q\end{array} \right)^{-1}\left(\begin{array}{c}
\bw_1 \\
\bw_2\end{array} \right).\]
This implies that one of $\bw_1$ and $\bw_2$ must be nonzero, otherwise $\bc$ is equal to zero. This completes the proof.
\end{proof}

The condition given in Lemma \ref{lem:2.1} can be converted to a condition on the coefficient matrix of the equation (\ref{eq:2.1}) as shown below.

\begin{thm}\label{thm:2.2} Let $A$ be the coefficient matrix of the equation (\ref{eq:2.1}). Then  the equation (\ref{eq:2.1}) has a nonzero solution in $\F^n_q$ if and only if $A^{(q)}$ and $A$ are row equivalent, where $A^{(q)}$ is obtained from $A$ by raising every entry to its $q$th power.
\end{thm}
\begin{proof} It is easy to see that $\bc^q$ is a solution of $A^{(q)}\bx^T=\bo$ whenever $\bc$ is a solution of (\ref{eq:2.1}) and vice versa. By Lemma \ref{lem:2.1}, this implies that  the equation (\ref{eq:2.1}) has a nonzero solution in $\F^n_q$ if and only if the equation $A^{(q)}\bx^T=\bo$ and the equation (\ref{eq:2.1}) have the same solution space, i.e., $A^{(q)}$ and $A$ are row equivalent.
\end{proof}

For any distinct elements $\Ga_1,\dots,\Ga_n$ of $\Fs$, we denote by $A(\Ga_1,\dots,\Ga_n)$ the matrix
\[\left(\begin{array}{cccc}
1 &1&\dots&1\\
\a_1&\a_2&\dots&\a_n\\
\a_1^2&\a_2^2&\cdots&\a_n^2\\
\vdots&\vdots&\ddots&\vdots\\
\a_1^{n-2}&\a_2^{n-2}&\cdots&\a_n^{n-2}\end{array} \right).\]

Now we are going to illustrate Theorem \ref{thm:2.2} by some explicit examples. In these examples, as long as we choose suitable $A$, then $A$ and $A^{(q)}$ are row equivalent and thus (\ref{eq:2.1}) has a solution over $\F_q$.

\begin{exm}\label{exm:2.3}{\rm Let $m$ be a divisor of $q^2-1$ and let $n=m+1$. Let $\Ga_1,\dots,\Ga_m$ be all the $m$-th roots of unity.  We claim that the system $A(0,\Ga_1,\dots,\Ga_m)\bx=\bo$ has a nonzero solution in $\F_q^n$. To prove this, it is sufficient to show that the rows of $A^{(q)}(0,\Ga_1,\dots,\Ga_m)$ are a permutation of the rows of $A(0,\Ga_1,\dots,\Ga_m)$.
The first row of the two matrices are identical. Hence, it is sufficient to show that the last $n-2=m-1$ rows of $A^{(q)}(0,\Ga_1,\dots,\Ga_m)$ are  a permutation of those of $A(0,\Ga_1,\dots,\Ga_m)$. To see this, we notice that the powers in the  last $m-1$ rows of $A^{(q)}(0,\Ga_1,\dots,\Ga_m)$ consist of $\{1\cdot q,2\cdot q,\dots,(m-1)\cdot q\}$, while  the powers in the last $m-1$ rows of $A(0,\Ga_1,\dots,\Ga_m)$ consist of $\{1,2\dots,m-1\}$. Thus, the desired result follows from
 the fact that the set $\{1\cdot q\pmod{m},2\cdot q\pmod{m},\dots,(m-1)\cdot q\pmod{m}\}$ and  the set $\{1,2\dots,m-1\}$ are identical.
}
\end{exm}

\begin{exm}\label{exm:2.4}{\rm Let $\beta$ be an element of $\Fs$, then, for any distinct elements $\Ga_1,\dots,\Ga_n$ of $\F_q$, it is easy to see that the matrices $A(\Gb+\Ga_1,\dots,\Gb+\Ga_n)$ and  $A(\Ga_1,\dots,\Ga_n)$ are row equivalent. As we have that $A^{(q)}(\Gb+\Ga_1,\dots,\Gb+\Ga_n)=A((\Gb+\Ga_1)^q,\dots,(\Gb+\Ga_n)^q)=A(\Gb^q+\Ga_1,\dots,\Gb^q+\Ga_n)$, the matrices $A^{(q)}(\Gb+\Ga_1,\dots,\Gb+\Ga_n)$ and  $A(\Ga_1,\dots,\Ga_n)$ are also row equivalent. This means that $A(\Gb+\Ga_1,\dots,\Gb+\Ga_n)$ is row equivalent to $A^{(q)}(\Gb+\Ga_1,\dots,\Gb+\Ga_n)$. By Theorem \ref{thm:2.2},
the system $A(\Gb+\Ga_1,\dots,\Gb+\Ga_n)\bx^T=\bo$ has a nonzero solution in $\F_q^n$.
}
\end{exm}

For any distinct elements $\Ga_1,\dots,\Ga_{n}$ of $\Fs$, we denote by $A_{\Pin}(\Ga_1,\dots,\Ga_{n})$ the $n\times (n+1)$ matrix
\[\left(\begin{array}{ccccc}
1 &1&\dots&1&0\\
\a_1&\a_2&\dots&\a_{n}&0\\
\a_1^2&\a_2^2&\cdots&\a_{n}^2&0\\
\vdots&\vdots&\ddots&\vdots&\vdots\\
\a_1^{n-2}&\a_2^{n-2}&\cdots&\a_{n}^{n-2}&0\\
\a_1^{n-1}&\a_2^{n-1}&\cdots&\a_{n}^{n-1}&1
\end{array} \right).\]

\begin{exm}\label{exm:2.5}{\rm Put $n=q^2$. Let $\Ga_1=0$ and let $\Ga_2,\dots,\Ga_{n}$ be all the $(q^2-1)$-th roots of unity.  We claim that the system $A_{\Pin}(\Ga_1,\dots,\Ga_{n})\bx^T=\bo$ has a nonzero solution in $\F_q^{n+1}$. To prove this, it is sufficient to show that the rows of $A_{\Pin}^{(q)}(\Ga_1,\dots,\Ga_{n})$ are a permutation of the rows of $A_{\Pin}(\Ga_1,\dots,\Ga_{n})$.  The first and  last rows of the two matrices are identical. Hence, it is sufficient to show that the middle $n-2=q^2-2$ rows of $A_{\Pin}^{(q)}(\Ga_1,\dots,\Ga_{n})$ are  a permutation of those of $A_{\Pin}(\Ga_1,\dots,\Ga_{n})$. To see this, we notice that the powers in the  middle $q^2-2$ rows of $A_{\Pin}^{(q)}(\Ga_1,\dots,\Ga_{n})$ consist of $\{1\cdot q,2\cdot q,\dots,(q^2-2)\cdot q\}$, while  the powers in the middle $q^2-2$ rows of $A_{\Pin}(\Ga_1,\dots,\Ga_{n})$ consist of $\{1,2\dots,q^2-2\}$. Thus, the desired result follows from
 the fact that the set $\{1\cdot q\pmod{q^2-1},2\cdot q\pmod{q^2-1},\dots,(q^2-2)\cdot q\pmod{q^2-1}\}$ and  the set $\{1,2\dots,q^2-2\}$ are identical.
}
\end{exm}

\section{Hermitian Self-Orthogonal $\GRS$ Codes}
Generally it is harder to construct Hermitian self-orthogonal codes  than Euclidean self-orthogonal codes due to raising of $q$th power . In this section, we provide a sufficient condition under which a generalized Reed-Solomon code is Hermitian self-orthogonal.

Let us first briefly introduce two inner products, i.e., Euclidean and Hermitian inner products.  Let $\F_q$ be the finite field of $q$ elements.
For two vectors $\bb=(b_1,\dots,b_n)$ and $\bc=(c_1,\dots,c_n)$ in
$\F^n_{q^2}$, we define the Euclidean inner product $\langle\bb,\bc\rangle_E$ and Hermitian inner product $\langle\bb,\bc\rangle_H$  to
be  $\langle\bb,\bc\rangle_E=\sum_{i=1}^nb_ic_i$ and $\langle\bb,\bc\rangle_H=\sum_{i=1}^nb_ic_i^q$, respectively.

For a $q^2$-ary linear code $C$ of length $n$, the Euclidean dual
$C^{\perp_E}$ (and the Hermitian dual $C^{\perp_H}$, respectively) of $C$ consists of all the vectors
 of $\F^n_{q^2}$ which are orthogonal to the codewords of $C$ with
 respect to the above Euclidean inner product (and the above Hermitian inner product, respectively).

For a vector $\bv=(v_1,\dots,v_n)\in\F^n_{q^2}$, let
$\bv^q=(v_1^q,\dots,v_n^q)$. For a subset $S$ of
$\F^n_{q^2}$, we define $S^{q}$ to be the set $\{\bv^q:\; \bv\in S\}$.
Then it is easy to verify that $V^{q}$ is also a subspace of
$\F^n_{q^2}$ whenever $V$ is a subspace of $\F^n_{q^2}$. It is easy to see that for a
$q^2$-ary linear code $C$, we have $C^{\perp_H}=(C^q)^{\perp_E}$. Therefore, $C$
is Hermitian self-orthogonal if and only if $C\subseteq
(C^q)^{\perp_E}$, i.e., $C^q\subseteq C^{\perp_E}$.

Next we review some notations and results of generalized Reed-Solomon
codes. The reader may refer  to \cite{LX04} for the details. Choose $n$
distinct elements $\a_1,\dots,\a_n$ of $\Fs$, and $n$ nonzero
elements $v_1,\dots,v_n$ of $\Fs$.

For an integer $k$ between $1$ and $n$, we denote by $\cP_k$ the set of polynomials in $\Fs[x]$ with degree at most $k-1$.  Define the code
\[\GRS_k(\ba,\bv):=\{(v_1f(\a_1),\dots,v_nf(\a_n)):\;f(x)\in\cP_k\}, \]
 where $\ba$ and $\bv$ denote the vectors
$(\a_1,\dots,\a_n)$ and $(v_1,\dots,v_n)$, respectively. The code $\GRS_k(\ba,\bv)$ is called a generalized  Reed-Solomon code  over $\Fs$. It is well known that a $\GRS$ code is a MDS code, i.e., attaining the Singleton bound saying that $k+d\leq n+1$.

Furthermore, we can define the GRS code by including the point at ``infinity" as the follows
{\footnotesize \[\GRS_k(\ba,\bv,\Pin):=\{(v_1f(\a_1),\dots,v_nf(\a_n),v_{n+1}f(\Pin)):\;f(x)\in\cP_k\}, \]}
where $\bv=(v_1,\dots,v_{n+1})\in(\F_{q^2}^*)^{n+1}$ and $f(\Pin)$ stands for the coefficient of $x^{k-1}$ in the polynomial $f(x)$. Then $\GRS_k(\ba,\bv,\Pin)$ is an MDS code of length $n+1$ and dimension $k$.

The Euclidean dual $\GRS_k(\ba,\bv)^{\perp_E}$  of $\GRS_k(\ba,\bv)$ is the
generalized RS code $\GRS_{n-k}(\ba,\bu)$, where $\bu$ is $(v_1^{-1}c_1,\dots, v_n^{-1}c_n)$ for a
nonzero solution $\bc=(c_1,\dots,c_n)$ of the system $A(\Ga_1,\dots,\Ga_n)\bx^T=\bo$.

In the same way, one can easily prove that the Euclidean dual $\GRS_k(\ba,\bv,\Pin)^{\perp_E}$  of $\GRS_k(\ba,\bv,\Pin)$ is the
generalized RS code $\GRS_{n+1-k}(\ba,\bu,\Pin)$, where $\bu$ is equal to $(v_1^{-1}c_1,\dots, v_{n+1}^{-1}c_{n+1})$ for a
nonzero solution $\bc=(c_1,\dots,c_{n+1})$ of the system $A_{\Pin}(\Ga_1,\dots,\Ga_n)\bx^T=\bo$.

Note that the solution space of the above system has dimension
$1$ since any $n-1$ columns ($n$ columns, respectively) of the above coefficient matrix form an invertible $(n-1)\times(n-1)$ matrix ($n\times n$ matrix, respectively). Thus, every nonzero solution has all coordinates not equal to
zero. Furthermore, the vector $\bu$ is uniquely determined up to a nonzero scalar.

\begin{lem}\label{lem:3.1} {\rm (i)}
Assume that the system $A(\Ga_1,\dots,\Ga_n)\bx^T=\bo$ has a nonzero solution $(c_1,\dots,c_n)\in\F_q^n$. Let $v_i^{q+1}=c_i$ for some $v_i\in\Fs$, $i=1,\dots,n$ (note that such $v_i$ always exist since $c_i\in\F_q$). If  $\GRS_k(\ba,\bv)^q\subseteq \GRS_{n-k}(\ba,\bv^q)$, then $\GRS_k(\ba,\bv)$ is Hermitian self-orthogonal.

{\rm (ii)} Assume that the system $A_{\Pin}(\Ga_1,\dots,\Ga_n)\bx^T=\bo$ has a nonzero solution $(c_1,\dots,c_n,c_{n+1})\in\F_q^{n+1}$. Let $v_i^{q+1}=c_i$ for some $v_i\in\Fs$, $i=1,\dots,n+1$ (note that such $v_i$ exist since $c_i\in\F_q$). If  $\GRS_k(\ba,\bv,\Pin)^q\subseteq \GRS_{n-k}(\ba,\bv^q,\Pin)$, then $\GRS_k(\ba,\bv,\Pin)$ is Hermitian self-orthogonal.
\end{lem}

\begin{proof} Let $\bv=(v_1,\dots,v_n)$. Then $\bv^q=(v_1^q,\dots,v_n^q)=(v_1^{-1}c_1,\dots,v_n^{-1}c_n)$. Hence, $\GRS_{n-k}(\ba,\bv^q)=\GRS_k(\ba,\bv)^{\perp_E}$. By the given condition, we have
the following inclusion $\GRS_k(\ba,\bv)^q
\subseteq\GRS_{n-k}(\ba,\bv^q)=\GRS_k(\ba,\bv)^{\perp_E}$.
 The desired result follows.

 The proof of the second part is very similar to the first part and we skip it.
\end{proof}

Based on Lemma \ref{lem:3.1}, we have the following simple result.
\begin{cor}\label{cor:3.2} {\rm (i)} Assume that the system $A(\Ga_1,\dots,\Ga_n)\bx^T=\bo$ has a nonzero solution $(c_1,\dots,c_n)\in\F_q^n$. Let $v_i^{q+1}=c_i$ for some $v_i\in\Fs$, $i=1,\dots,n$. If $q(k-1)+1\le n-k$, i.e., $k\le (n+q-1)/(q+1)$, then $\GRS_k(\ba,\bv)$ is Hermitian self-orthogonal.

{\rm (ii)} Assume that the system $A_{\Pin}(\Ga_1,\dots,\Ga_n)\bx^T=\bo$ has a nonzero solution $(c_1,\dots,c_n,c_{n+1})\in\F_q^{n+1}$. Let $v_i^{q+1}=c_i$ for some $v_i\in\Fs$, $i=1,\dots,n+1$. If $q(k-1)+1= n+1-k$, i.e., $k= (n+q)/(q+1)$, then $\GRS_k(\ba,\bv,\Pin)$ is Hermitian self-orthogonal.
\end{cor}
\begin{proof} (i)
The desired result follows from  the following inclusions
\[\GRS_k(\ba,\bv)^q\subseteq \GRS_{q(k-1)+1}(\ba,\bv^q)
\subseteq\GRS_{n-k}(\ba,\bv^{q})\]
and Lemma \ref{lem:3.1}(i).

(ii) The second part follows from the inclusion
\small \[\GRS_k(\ba,\bv,\Pin)^q\subseteq \GRS_{q(k-1)+1}(\ba,\bv^q,\Pin)\]
and Lemma \ref{lem:3.1}(ii).
\end{proof}

\begin{exm}\label{exm:3.3}{\rm Let $n=q^2$. Let $\Fs=\{\Ga_1,\dots,\Ga_n\}$. By Example \ref{exm:2.3}, there is a nonzero solution $\bc=(c_1,\dots,c_n)\in\F_q^n$ of the system $A(\Ga_1,\dots,\Ga_n)\bx^T=\bo$. Let $v_i\in\Fs$ such that $v_i^{q+1}=c_i$ for every $1\le i\le n$. By Corollary \ref{cor:3.2}, $\GRS_k(\ba,\bv)$ is Hermitian self-orthogonal for any  $k\le q-1$. This result was obtained in \cite{Jin1}.
}\end{exm}

\begin{exm}\label{exm:3.4}{\rm Let $n=\sum_{j=1}^tn_j$ for some positive integers $n_j\le q$ and $t\le q$. Write $\F_q=\{\Ga_1,\dots,\Ga_q\}$ and  $\Fs$ into the union of additive cosets $\cup_{j=1}^q(\Gb_j+\F_q)$ for some $\Gb_j\in\Fs$.  By Example \ref{exm:2.4}, for each $j$, there exists a nonzero solution $\bc_j=(c_{j1},\dots,c_{jn_j})\in\F_q^{n_j}$ of $A(\Gb_j+\Ga_1,\dots,\Gb_j+\Ga_{n_j})\bx^T=\bo$. Let $v_{ji}\in\Fs$ such that $v_{ji}^{q+1}=c_{ji}$ for all $1\le j\le t$ and $1\le i\le n_j$. Then it is easy to see that  $\GRS_k(\bb_j,\bv_j)^q=\GRS_k(\ba_j,\bv_j^q)$, where $\bb_j=(\Gb_j+\Ga_1,\dots,\Gb_j+\Ga_{n_j})$, $\ba_j=(\Ga_1,\dots,\Ga_{n_j})$ and $\bv_j=(v_{j1},\dots,v_{jn_j})$. Therefore, by Lemma \ref{lem:3.1}, $\GRS_k(\bb_j,\bv_j)$ is Hermitian self-orthogonal for any $k\le n_j/2$. Combining all codes $\GRS_k(\bb_j,\bv_j)$, we obtain a Hermitian self-orthogonal code  $\GRS_k(\bb,\bv)$ of length $n$, where $\bb=(\bb_1,\dots,\bb_t)$ and $\bv=(\bv_1,\dots,\bv_t)$ for any $k\le \min\{n_1,\dots,n_t\}/2+1$. This code was constructed in \cite{Jin1} as well.
}\end{exm}

\begin{rem}\label{rem:3.5}{\rm
 All other  Hermitian self-orthogonal GRS codes given in \cite{Jin1} can be obtained using Lemmas \ref{lem:2.1} and \ref{lem:3.1}.
 }\end{rem}

The next theorem shows existence of $q^2$-ary Hermitian self-orthogonal MDS codes of length $q^2+1$ and dimension up to $q$.

\begin{thm}\label{thm:3.5} There exists a  $q^2$-ary Hermitian self-orthogonal MDS code of length $q^2+1$ and dimension $q$.
\end{thm}
\begin{proof} Let $n=q^2$ and let $\Ga_1,\dots,\Ga_n$ be all $q^2$ elements of $\Fs$. By Example \ref{exm:2.5}, there is a nonzero solution $\bc=(c_1,\dots,c_{n+1})\in\F_q^{n+1}$ of the system $A_{\Pin}(\Ga_1,\dots,\Ga_n)\bx^T=\bo$. Let $v_i\in\Fs$ such that $v_i^{q+1}=c_i$ for every $1\le i\le n+1$. By Corollary \ref{cor:3.2}(ii), $\GRS_q(\ba,\bv,\Pin)$ is Hermitian self-orthogonal.
\end{proof}

In the following lemma, we give a sufficient condition for inclusion of certain GRS codes.
\begin{lem}\label{lem:3.6} Let  $\Ga_1,\dots,\Ga_n$ be $n$ distinct roots of $x^n-x$ in $\Fs$ and put $\ba=(\Ga_1,\dots,\Ga_n)$. If the set $\{q i\pmod{n-1}:\; 1\le i\le k-1\}$ is a subset of $\{0,1,\dots,n-k-1\}$, then $\GRS_k(\ba,\bv)^q\subseteq \GRS_{n-k}(\ba,\bv^q)$ for any $\bv\in (\F_{q^2}^*)^n$.
\end{lem}
\begin{proof} It is clear that $\GRS_k(\ba,\bv)^q$ has a basis $\{(v_1^q,\dots,v_n^q)\}\cup\{(v_1^q\Ga_1^{qi},\dots,v_n^q\Ga_n^{qi})\}_{i=1}^{k-1}$. Hence, it is sufficient to show that $(v_1^q\Ga_1^{qi},\dots,v_n^q\Ga_n^{qi})$ belongs to $\GRS_{n-k}(\ba,\bv^q)$ for all $1\le i\le k-1$. Assume that $qj\pmod{n-1}= \ell\in\{1,\dots,n-k-1\}$. Then we have $(v_1^q\Ga_1^{qi},\dots,v_n^q\Ga_n^{qi})=(v_1^q\Ga_1^{\ell},\dots,v_n^q\Ga_n^{\ell})\in  \GRS_{n-k}(\ba,\bv^q)$. This completes the proof.
\end{proof}

\begin{thm}\label{thm:3.7}
Let $r$ be an integer satisfying $0\leq r< q+1$ and $q+1\equiv r\mod 2r$.  Then for any $k\leq (q-1+r)/2$, there exists a $q^2$-ary Hermitian self-orthogonal GRS code of length $n=r(q-1)+1$ and dimension $k$.
\end{thm}
\begin{proof} Put $m=r(q-1)$. By the condition that  $q+1\equiv r\mod 2r$, it is easy to see that $r$ divides $(q+1)$. Thus, $m$ divides $q^2-1$. Let $\Ga_1=0$ and let $\Ga_2,\dots,\Ga_n$ be all $m$th roots of unity in $\Fs$. Then by Example \ref{exm:2.3}, there exists a nonzero solution $\bc\in \F_q^n$ of the system $A(\Ga_1,\dots,\Ga_n)\bx^T=\bo$. Let $\bv\in\F_{q^2}^n$ such that $\bv^{q+1}=\bc$.

Now by Lemmas \ref{lem:3.6} and \ref{lem:3.1}, it is sufficient to show that $\{q j\pmod{m}:\; 1\le j\le k-1\}$ is a subset of $\{0,1,\dots,n-k-1\}$.
For $1\leq j\leq k-1$, write $j\equiv a\mod r$ with $0\le a\le r-1$. Then
$jq \mod{m}= aq+j-a.$

If $0\leq a\leq r-2$, then it is easy to check that \[qj\mod{m}=aq+j-a\leq r(q-1)-k= n-k-1.\]
 Now we  consider the case where $a=r-1$.
In this case, we have $({q-r-1})/{2}\equiv r-1\mod r$. This implies that   for any number  $\ell$ satisfying $(q-r-1)/{2}<\ell<(q-1+r)/2$, one must have $\ell\mod{r}\le r-2$. Thus, we must have $j\le \frac{q-r-1}{2}$. Therefore,

\[\begin{array}{ll}
jq\mod{m}&= (r-1)q+j-r+1 =m-(q-1)+j\\ &\le m-\frac{q+r-1}2\le m-k=n-k-1.
\end{array}\]
This completes the proof.
\end{proof}

\begin{cor}\label{cor:3.8} If $3$ divides $q+1$, then for any $k\le (2q-1)/3$, there exists a $q^2$-ary Hermitian self-orthogonal GRS code of length $n=(q^2+2)/3$ and dimension $k$.
\end{cor}
\begin{proof} Let $r=(q+1)/3$, i.e., $q+1=3r$. Then $q+1\equiv r\mod{2r}$, $n=r(q-1)+1=(q^2+2)/3$ and $(q+r-1)/2=(2q-1)/3$. The desired result follows from Theorem \ref{thm:3.7}.
\end{proof}

\begin{exm}\label{exm:3.9} {\rm \begin{itemize}\item[(i)] If $q+1\equiv 3\mod{6}$, then by Theorem \ref{thm:3.7}, there a $q^2$-ary Hermitian self-orthogonal GRS code of length $n=3(q-1)+1$ and dimension $k$  for any $k\le (q+2)/2$.
\item[(ii)] If $q+1\equiv 8\mod{16}$, then by Theorem \ref{thm:3.7}, there is a $q^2$-ary Hermitian self-orthogonal GRS code of length $n=8(q-1)+1$ and dimension $k$  for any $k\le (q+7)/2$.
\end{itemize}
}\end{exm}
Finally, we employ the language of algebraic geometry codes to show existence of $q^2$-ary Hermitian self-orthogonal MDS code of length $q^2+1$ and dimension $k$ for any $1\le k\le q-2$. The reader may refer to \cite{stich} for the details on algebraic geometry codes.

\begin{thm}\label{thm:3.11} Let $q\ge 3$. Then for any $1\le k\le q-2$, there exists a  $q^2$-ary Hermitian self-orthogonal MDS code of length $q^2+1$ and dimension $k$.
\end{thm}
\begin{proof} We label elements of $\F_{q^2}$ as $\Ga_1,\Ga_2,\dots,\Ga_{q^2}$. Let $P_i$ be the  zero of $x-\Ga_i$ for $1\le i\le q^2$ and let $P_{q^2+1}$ be the pole of $x$. For a divisor $A$ with ${\rm supp}(A)\cap\{P_1,P_2,\dots,P_{q^2},P_{q^2+1}\}=\emptyset$ and nonzero elements $v_1,v_2,\dots,v_{q^2+1}$ of $\F_{q^2}$, we define the algebraic geometry code
\[\begin{split} &C_L(D,A,\bv):=\\
&\{(v_1f(P_1),v_2f(P_2),\dots,v_{q^2+1}f(P_{q^2+1})):\; f\in\mathcal{L}(A)\},\end{split}\]
where $D$ is the divisor $\sum_{i=1}^{q^2+1}P_i$ and $\bv=(v_1,v_2,\dots,v_{q^2+1})$. The Euclidean dual of $C_L(D,A,\bv)$ is given by
\[\begin{split} & C_\Omega(D,A,\bv^{-1}):=\\
&\{(v_1^{-1}{\rm res}_{P_1}(\eta),,\dots,v_{p^2+1}^{-1}{\rm res}_{P_{q^2+1}}(\eta):\; \eta\in\Omega(G-A)\}.\end{split}\]

Choose monic irreducible polynomials $a(x), b(x)$ and $c(x)$ over $\F_{q^2}[x]$ of degree $k+1$, $2$ and $q-k$, respectively. As $q-k\ge 2$, we have $\gcd(a(x)b(x)c(x),x^{q^2}-x)=1$, i.e., $\Ga$ are not roots of $a(x)b(x)c(x)$ for all $\Ga\in\F_{q^2}$. Now let $G$ be the divisor $(a(x))_0-(b(x))_0$, where $(f(x))_0$ stands for the zero divisor of a nonzero polynomial $f(x)$. 
Then we have $\deg(G)=k-1$.

Define the differential
\[\eta=\left(\frac{a(x)c(x)}{b(x)}\right)^{q+1}\times \frac{1}{x^{q^2}-x}dx.\]
Then we have
\[{\rm res}_{P_i}(\eta)=-\left(\frac{a(\Ga_i)c(\Ga_i)}{b(\Ga_i)}\right)^{q+1}\in\F_q^*\]
for $1\le i\le q^2$ and ${\rm res}_{P_{q^2+1}}(\eta)=-1$.

Put $u_i={\rm res}_{P_i}(\eta)$ for all $1\le i\le q^2+1$. Then, for each $1\le i\le q^2+1$, there exists $v_i\in \F_{q^2}^*$ such that $v_i^{q+1}=u_i$ since $u_i\in\F_q^*$. By \cite[Theorem I.5.14]{stich}, there exists an $\F_q$-isomorphism $\pi$: $\mathcal{L}(D-G+(\eta))\rightarrow \Omega_F(G-D)$ given by $\pi(f)=f\eta$. Furthermore, for $f\in \mathcal{L}(D-G+(\eta))$, we have ${\rm res}_{P_i}(f\eta)=f(P_i){\rm res}_{P_i}(\eta)$. This implies that $C_\Omega(D,G,\bi)=C(D, D-G+(\eta),\bu)$, where $\bi$ stands for the all-one vector.
This gives that
$C^{\perp_E}_L(D,G,\bv)=C_\Omega(D,G,\bv^{-1})=C_L(D,D-G+(\eta),\bu\bv^{-1})$, where $\bu\bv^{-1}$ stands for the vector $(u_1/v_1,\dots,u_{q^2+1}/v_{q^2+1})$. Hence, $\bu\bv^{-1}=\bv^q$.

As $(\eta)=(q+1)G+(q+1)(c(x))_0-D$, we have $D-G+(\eta)=qG+(q+1)(c(x))_0$. Thus, we get $qG\le D-G+(\eta)$. This implies that
\begin{eqnarray*}
C^{q}_L(D,G,\bv)&\subseteq& C_L(D,qG,\bv^q)=C_L(D,qG,\bu\bv^{-1})\\
&\subseteq &C_L(D,qG+(q+1)(c(x))_0,\bu\bv^{-1})\\
&=&C_L(D,D-G+(\eta),\bu\bv^{-1})\\
&=&C^{\perp_E}_L(D,G,\bv),
\end{eqnarray*}
i.e., $C_L(D,G,\bv)$ is a Hermitian self orthogonal MDS code.



\end{proof}

\begin{rem}\label{rem:3.6}{\rm
Theorems \ref{thm:3.5} and \ref{thm:3.11} extend the result given in \cite{Jin1}. Note that in \cite{Jin1}, only $q^2$-ary Hermitian self-orthogonal MDS codes of length $q^2+1$ and dimension at most $q/2$ were produced.
 }\end{rem}

\section{Quantum codes}

Let us first recall a connection between classical Hermitian self-orthogonal MDS codes and quantum MDS codes.
\begin{lem}(\cite{Ash Kni,Jin1})\label{lem:4.1} There is a $q$-ary
$[[n,n-2k,k+1]]$ quantum MDS code whenever there exists a $q^2$-ary
classical Hermitian self-orthogonal  $[n,k,n-k+1]$-MDS code.
\end{lem}

By combining Theorems \ref{thm:3.5}, \ref{thm:3.11}  and Lemma \ref{lem:4.1}, we can immediately get the quantum MDS codes of length $q^2+1$ in the following theorem.
\begin{thm}\label{thm:4.1} For any $k\in\{1,2,\dots,q-2,q\}$, there exists a $q$-ary $[[q^2+1, q^2+1-2k,k+1]]$ quantum MDS code.
\end{thm}

\begin{rem}\label{rem:4.2}{\rm In \cite{Jin1}, $q$-ary $[[q^2+1, q^2+1-2k,k+1]]$ quantum MDS codes were constructed only for $k\le q/2$. The above Theorem also extends the results given in \cite{Gu11,KZ12}.
}\end{rem}

By combining Theorem \ref{thm:3.7} and Corollary \ref{cor:3.8} in the previous section and Lemma \ref{lem:4.1}, we can immediately get the quantum codes in the following theorem.
\begin{thm}\label{thm:4.2} Let $r$ be an integer satisfying $0\leq r< q+1$ and $q+1\equiv r\mod 2r$. Then for any $k\leq (q-1+r)/2$, there exists a $q$-ary $[[r(q-1)+1, r(q-1)+1-2k,k+1]]$ quantum MDS code. In particular, if $3$ divides $q+1$, then for any $k\le (2q-1)/3$, there exists a $q$-ary $[[(q^2+2)/3, (q^2+2)/3-2k,k+1]]$ quantum MDS code.
\end{thm}

\begin{rem}\label{rem:4.3}{\rm Except for a few sparse lengths $n$ such as $n=(q^2+1)/2$ (see \cite{KZ12}), $q^2$ and $q^2+1$ (see \cite{Gra Bet2,Jin1,Li Xin Wan2}) and some other lengths derived from propagation rule, all known $q$-ary quantum MDS codes of other lengths have minimum distance at most $q/2+1$. Our Theorem \ref{thm:4.2} provides some new $q$-ary quantum codes with minimum distance bigger than $q/2+1$. In particular, we construct a $q$-ary quantum code of length $(q^2+2)/3$ with minimum distance up to $(2q-1)/3$.
}\end{rem}
In the following examples, using our result in this paper, we can produce quantum codes which include some of the previously known results and some new quantum MDS codes.

\begin{exm}\label{exm:4.4}{\rm \begin{itemize}
\item[(i)] If $q\equiv2\mod 6$, we can derive a $q$-ary $[[3q-2,3q-2-2k,k+1]]$ quantum MDS code for any $k\leq(q+2)/2$ from Example \ref{exm:3.9} and Lemma \ref{lem:4.1}.

In particular, we can produce a $[[22,12,6]]_8$ quantum MDS code. This code is new.
\item[(ii)] If $q\equiv 7\mod 16$, we can derive a $q$-ary $[[8q-7,8q-7-2k,k+1]]$ quantum MDS code for any $k\leq(q+7)/2$ from Example \ref{exm:3.9} and Lemma \ref{lem:4.1}.

In particular, we can produce a $23$-ary $[[177,147,16]]$ quantum MDS code which is also a new code.
\end{itemize}
}\end{exm}

\begin{exm}\label{exm:4.5}{\rm By Example \ref{exm:3.3} and Lemma \ref{lem:4.1}, we can drive a $q$-ary $[[q^2,q^2-2k,k+1]]$ quantum MDS code for any $k\le q-1$. This result was obtained in \cite{Jin1}.
}\end{exm}

\begin{exm}\label{exm:4.6}{\rm Let $n=\sum_{j=1}^tn_j$ for some positive integers $n_j\le q$ and $t\le q$. By Example \ref{exm:3.4} and Lemma \ref{lem:4.1}, we can derive a $q$-ary $[[n,n-2k,k+1]]$ quantum MDS code for any $k\le \min\{n_1,\dots,n_t\}/2+1$. This result was obtained in \cite{Jin1}.
}\end{exm}



\begin{thebibliography}{99}
\bibitem{Aly Kla Sar} S. A.~Aly, A.~Klappenecker and
P. K.~Sarvepalli, ``On quantum and classical BCH codes,"
\emph{IEEE Trans. Inform. Theory,} vol. 53, no. 3, pp. 1183--1188,
2007.
\bibitem{Ash Kni} A.~Ashikhmin and E.~Knill, ``Nonbinary quantum stablizer
codes," \emph{IEEE Trans. Inform. Theory,} vol. 47, no. 7, pp.
3065--3072, 2001.


\bibitem{Bie Ede} J.~Bierbrauer and Y.~Edel, ``Quantum twisted
codes," \emph{J. Comb. Designs,}  vol. 8, pp. 174--188, 2000.

\bibitem{Cal Rai}A. R.~Calderbank, E. M.~Rains, P. W.~Shor and
N. J. A.~Sloane, ``Quantum error correction via codes over GF(4),"
\emph{IEEE Trans. Inform. Theory,} vol. 44, no. 4, pp. 1369--1387,
 1998.

\bibitem{Che Lin Xin} H.~Chen, S.~Ling and C.~Xing, ``Quantum codes from concatenated
algebraic-geometric codes," \emph{IEEE Trans. Inform. Theory,} vol. 51, no. 8, pp. 2915--2920,
2005.

\bibitem{Feng} K.~Feng, ``Quantum code $[[6,2,3]]_p$ and $[[7,3,3]]_p$ $(p\geq 3)$
exists," \emph{IEEE Trans. Inform. Theory,} vol. 48, no. 8, pp.
2384--2391,  2002.

\bibitem{FLX06} K. Feng, S. Ling, C. Xing, ``Asymptotic bounds on quantum codes from algebraic geometry codes," \emph{IEEE Trans. on Inf. Theory,} vol. 52, no. 3, pp.986-991, 2006.





\bibitem{Gra Bet2} M.~Grassl, T.~Beth and M.~R\"{o}ttler, ``On optimal quantum
codes," \emph{Int. J. Quantum Inf.,} vol. 2, no. 1, pp. 757--775,
2004.

\bibitem{Gra2} M.~R\"{o}ttler, M.~Grassl and T.~Beth,``On quantum MDS codes," In Proc. Int. Symp. Inform. Thoory, Chicago, USA, pp. 356, 2004.

\bibitem{Gu11}  G. G. L. Guardia, ``New quantum MDS codes," \emph{IEEE Trans. Inform.
Theory,} vol. 57, no. 8, pp. 5551-554,  2011.

\bibitem{Hu Tan} D.~Hu, W.~Tang, M.~Zhao, Q.~Chen, S.~Yu and C. H.~Oh, ``Graphical nonbinary quantum error-correcting codes,"
\emph{Phy. Rev. A.,} vol. 78, no. 1, 012306,  2008.

\bibitem{Jin1} L. F. Jin, S. Ling, J. Q. Luo and C. P. Xing, ``Application of Classical Hermitian Self-orthogonal MDS Codes to Quantum MDS
Codes," \emph{IEEE. Trans. Inform. Theory,} vol. 56, pp.4735-4740, 2010.



\bibitem{KZ12} X. Kai and S. Zhu,``New quantum MDS codes from negacyclic codes,"
 \emph {IEEE Trans. Inform. Theory}, vol. 59, no. 2,  pp. 1193-1197, 2012.





\bibitem{Ket Kla} A.~Ketkar, A.~Klappenecker, S.~Kumar and P.~Sarvepalli, "Nonbinary stablizer codes over finite fields," \emph{IEEE.
Trans. Inform. Theory,} vol. 52, no. 11,  pp.4892--4914, Nov.
2006.





\bibitem{Li Xin Wan2} Z.~Li, L. J.~Xing and X. M.~Wang, ``Quantum generalized Reed-Solomon codes: unified
framework for quantum MDS codes,"  \emph{Phys. Rev. A,} vol. 77, pp. 012308-1--12308-4, 2008.


\bibitem{Li Xu} R.~Li and Z.~Xu, ``Construction of $[[n,n-4,3]]_q$ quantum MDS codes for odd prime power $q$,"
\emph{Phys. Rev. A,} vol. 82, pp.052316-1--052316-4, 2010.

\bibitem{LX04} S. Ling and C. P. Xing, {\it Coding Theory -- A
First Course,} Cambridge University Press, 2004.

\bibitem{Rain} E. M.~Rains, ``Nonbinary quantum codes," \emph{IEEE. Trans. Inform. Theory,} vol. 45, no. 6,
pp.1827--1832, Sept. 1999.

\bibitem{Sar Kla} P. K.~Sarvepalli and A.~Klappenecker, ``Nonbinary quantum Reed-Muller
codes," \emph{Proc. Inter. Sym. Inf. Theory 2005,} Adelaide,
Australia, pp. 1023--1027.

\bibitem{Shor} P. W.~Shor, ``Scheme for reducing decoherence in quantum computer
memory," \emph{Phys. Rev. A,} vol. 52, no. 4,  pp. R2493--R2496,
1995.

\bibitem{stich} H. Stichtenoth, {\it Algebraic function fields and codes}, Springer Verlag, 2003.

\bibitem{Stea3} A. M.~Steane, ``Enlargement of Calderbank-Shor-Steane quantum
codes," \emph{IEEE. Trans. Inform. Theory,} vol. 45, no. 7, pp.
2492--2495, Nov. 1999.



\end{thebibliography}
\end{document}